\documentclass[letterpaper, 10 pt, conference]{ieeeconf}  
\IEEEoverridecommandlockouts 
\usepackage{amsmath,amssymb,color,cite}
\usepackage{graphicx}

\usepackage{latexsym}
\usepackage{algorithm}
\usepackage{algpseudocode}
\usepackage{verbatim}
\usepackage{cite}
\usepackage[english]{babel}
\usepackage{balance}

\usepackage{amsmath,amsthm}

\let\labelindent\relax
\usepackage{enumitem}
\usepackage{xcolor}

\usepackage{pgf}
\usepackage{pgfplots,tikz}
\usetikzlibrary{tikzmark}
\usepackage[font=footnotesize]{caption}
\usepackage{arydshln}

\usepackage{amsfonts}
\usepackage{amssymb}
\usepackage{amsbsy}
\usepackage{bm}
\usepackage{mathtools}
\usepackage[
  colorlinks=true,
  linkcolor=black,
  urlcolor=blue,
  citecolor=black
]{hyperref}
\usepackage[font=footnotesize]{subcaption}
\usepackage[font=footnotesize]{caption}
\usepackage{yfonts}

\theoremstyle{definition}
\newtheorem{definition}{Definition}

\theoremstyle{plain}
\newtheorem{assumption}{Assumption}
\newtheorem{lemma}{Lemma}
\newtheorem{theorem}{Theorem}

\newcommand{\EE}{\mathbb{E}}
\newcommand{\FF}{\mathcal{F}}
\newcommand{\xx}{\mathbf{x}}
\newcommand{\XX}{\mathcal{X}}
\newcommand{\NN}{\mathcal{N}}
\newcommand{\RR}{\mathbb{R}}
\newcommand{\GG}{\mathcal{G}}
\newcommand{\vv}{\mathbf{v}}
\newcommand{\xxstar}{\mathbf{x^*}}
\newcommand{\xstar}{x^*}
\newcommand{\Xstar}{\mathcal{X}^*}
\newcommand{\MM}{\mathbf{M}}
\newcommand{\bfrak}{\mathfrak{b}}

\newcommand{\zz}{\mathbf{z}}
\newcommand{\UU}{\mathbf{U}}
\newcommand{\ee}{\mathbf{e}}

\allowdisplaybreaks

\title{\LARGE \bf
Last-Iterate Guarantees for Learning in Co-coercive Games
}

\author{Siddharth Chandak$^{1\dagger}$, Ramanan Tamizholi$^{2\dagger}$ and Nicholas Bambos$^{3}$
\thanks{$\dagger$ Equal contribution. Listed alphabetically.}
\thanks{$^{1,3}$ Department of Electrical Engineering, Stanford University, USA}
\thanks{$^{2}$ {Indian Institute of Science, Bengaluru, India}}
\thanks{$^{1}${\tt\small chandaks@stanford.edu}}
\thanks{$^{2}${\tt\small ramanant@iisc.ac.in}}
\thanks{$^{3}${\tt\small bambos@stanford.edu}}  
}

\begin{document}

\maketitle
\thispagestyle{empty}
\pagestyle{empty}

\begin{abstract}
We establish finite-time last-iterate guarantees for vanilla stochastic gradient descent in co-coercive games under noisy feedback. This is a broad class of games that is more general than strongly monotone games, allows for multiple Nash equilibria, and includes examples such as quadratic games with negative semidefinite interaction matrices and potential games with smooth concave potentials. Prior work in this setting has relied on relative noise models, where the noise vanishes as iterates approach equilibrium, an assumption that is often unrealistic in practice. We work instead under a substantially more general noise model in which the second moment of the noise is allowed to scale affinely with the squared norm of the iterates, an assumption natural in learning with unbounded action spaces. Under this model, we prove a last-iterate bound of order $O(\log(t)/t^{1/3})$, the first such bound for co-coercive games under non-vanishing noise. We additionally establish almost sure convergence of the iterates to the set of Nash equilibria and derive time-average convergence guarantees.

\end{abstract}
\section{Introduction}
Decision-making in multi-agent systems arises in many applications, including communication networks, distributed control, and economic markets. In these settings, each agent's outcome depends on both its own action and the actions of others, making strategic interaction a central feature. Game theory provides a natural mathematical framework for analyzing these interactions, where each agent seeks to maximize a utility function that captures its individual objective \cite{alpcan2009control, marden2013distributed}. In particular, \emph{continuous games}, where agents select actions from continuous spaces, are especially relevant for modeling real-valued decision-making problems such as resource allocation, pricing, and distributed learning.

An important question in such games is whether agents can learn a Nash equilibrium (NE) through repeated decentralized interaction. In continuous games, gradient-based methods provide a natural approach to decentralized learning due to their simplicity and scalability, as each agent updates its action using only local gradient information of its utility function \cite{tatarenko2020geometric, mazumdar2020gradient}. In many practical settings, however, this gradient feedback is noisy due to factors such as observational errors, random disturbances, or stochasticity in the underlying utility functions \cite{Zhou_2019, Chandak_learning}. The most natural learning rule in this setting is vanilla stochastic gradient descent (SGD), where at each step every agent updates its action by taking a step in the direction of the noisy gradient estimate of its utility function, without any additional modifications such as momentum, variance reduction, or extrapolation.

To establish convergence guarantees for learning dynamics, it is necessary to impose structural assumptions on the underlying game. Most existing convergence and finite-time analyses for stochastic gradient descent have focused on \emph{strongly monotone games} \cite{rosen1965existence}. Strong monotonicity of the game gradient operator enables an $\mathcal{O}(1/k)$ mean-square last-iterate convergence rate for vanilla SGD \cite{zhou2021robust}, and moreover guarantees uniqueness of the Nash equilibrium. However, the strong structural assumptions required and the uniqueness of equilibrium considerably restrict its scope of applicability.

In this work, we consider \emph{co-coercive games}, a broader class that has attracted recent interest in the literature \cite{Zhou_2020,loizou2021stochastic}. Co-coercivity is a substantially weaker condition on the game gradient operator than strong monotonicity, making convergence analysis for stochastic gradient methods considerably harder. Moreover, co-coercive games may admit non-singleton Nash equilibrium sets, introducing an additional challenge absent in the strongly monotone setting. This class includes several important examples, such as quadratic games with negative semidefinite interaction matrices and potential games with smooth concave potentials, which we present as motivating examples in the paper.

In this work, we establish the first finite-time last-iterate convergence guarantee for co-coercive games under stochastic gradient dynamics with non-vanishing noise. Specifically, for the vanilla SGD algorithm, we prove a last-iterate mean-square bound of order $\mathcal{O}\left(\log(t)/t^{1/3}\right).$ In contrast to prior works on co-coercive games, which rely on relative noise models where the noise vanishes as iterates approach equilibrium \cite{Zhou_2020}, our analysis is carried out under a general noise model in which the second moment of the noise scales affinely with the squared norm of the current iterate. This assumption is natural in learning with unbounded action spaces, and strictly subsumes the relative noise model. We additionally establish almost sure convergence of the iterates to the set of Nash equilibria and derive time-average convergence guarantees. Our analysis draws on techniques from the study of inexact Krasnosel'skii-Mann iterations \cite{Davis-splitting, Liang-inexact}, extended with additional arguments to handle last-iterate convergence under the general noise model.

\subsection{Related Work}
Learning in games has been studied extensively across economics, control theory, and machine learning. A central question in this literature is whether players converge to a Nash equilibrium when they update their strategies using only local information. Most convergence results for gradient-based methods rely on \emph{strong monotonicity}, which guarantees the existence of a unique equilibrium and has led to a well-developed theory with finite-time guarantees in both deterministic and stochastic settings (see, e.g., \cite{bravo2018bandit, Zhou_2019, tatarenko2022rate}). A related line of work focuses on generalized Nash equilibrium (GNE) problems, in which agents are coupled not only through their objective functions but also through shared constraints \cite{Chandak_learning, tatarenkoGNE}. Establishing convergence in this setting is generally more challenging, owing to the additional complexity introduced by the constraint structure.

Beyond strong monotonicity, works have also considered \emph{monotone games}, where the game gradient operator is merely monotone. This setting is more challenging for vanilla SGD, as the contraction properties exploited under strong monotonicity are no longer available. Indeed, it is known that vanilla SGD can fail to converge in this setting, with iterates exhibiting cycling or divergence even in simple examples \cite{mertikopoulos2018cycles}. As a result, guarantees in this setting have largely been obtained for modified algorithms such as optimistic gradient descent and extragradient methods (see, e.g., \cite{gorbunov2022last, tatarenko2019monotone, cai2022tight}).

Between strongly monotone and monotone games lies the class of \emph{co-coercive games}, first introduced in \cite{Zhou_2020}. In this setting, the game gradient operator is co-coercive, providing more structure than monotonicity while being less restrictive than strong monotonicity. The same work also established last-iterate convergence guarantees for vanilla SGD in this setting, although their analysis relies on a relative noise model where the noise vanishes as the iterates approach equilibrium. In contrast, we consider a more relaxed noise assumption (see Section~\ref{sec:SGD_noise}). A related notion, termed \emph{expected co-coercivity}, was introduced in \cite{loizou2021stochastic}. This property is implied by quasi-strong monotonicity, a condition defined relative to a fixed point $\xxstar$. Similar to strong monotonicity, quasi-strong monotonicity guarantees a unique Nash equilibrium and enables linear convergence rates.

\subsection{Outline and Notation}
The rest of the paper is organized as follows. Section~\ref{sec:coco_games} introduces co-coercive games and presents two motivating examples. Section~\ref{sec:SGD_noise} describes the vanilla SGD update rule and the assumptions on the noisy feedback received by players. Section~\ref{sec:results} presents the main results and Section~\ref{sec:outline} outlines the proof strategy. Finally, Section~\ref{sec:conc} concludes with directions for future work.

We use bold letters to denote vectors and the standard game-theoretic notation of $\xx_{-n}$ to denote the vector of actions for all players except player $n$. $\|\cdot\|$ denotes the Euclidean norm throughout the paper and $\langle\xx,\zz\rangle=\xx^\top\zz$ denotes the inner product between vectors $\xx$ and $\zz$. 

\section{Co-coercive Games}\label{sec:coco_games}
Consider a set of players $\NN=\{1,2,\ldots, N\}$, where each player $n$ chooses an action $x_n=(x_n^{(1)}, \ldots, x_n^{(d)})\in\RR^d$. Let $\xx=(x_1,\ldots,x_N)\in\XX=\RR^{Nd}$ denote the concatenation of all action vectors. A continuous game is defined as follows.
\begin{definition}
    A continuous game is a tuple $\GG=(\NN, \XX, \{u_n\}_{n=1}^N)$ where $\NN$ is the set of $N$ players, $\XX=\RR^{Nd}$ is the set of action profiles and $u_n:\XX\mapsto\RR$ is the $n$-th player's utility function.
\end{definition}
The properties of the game are characterized through the gradient operator induced by the utility functions. We assume that the function $u_n(\xx)$ is continuously differentiable in $\xx$ for each $n\in\NN$. Our key assumption is that the game is co-coercive. For this, we define the individual gradient as $v_n(\xx)=\nabla_{x_n} u_n(\xx)$, 
and the game gradient operator as
$$\vv(\xx)=(v_1(\xx),\ldots, v_N(\xx)).$$
\begin{definition}\label{defn:co-coercive}
    For $\lambda>0$, a continuous game with continuously differentiable utility functions is called $\lambda$-co-coercive if the following holds for all $\xx,\xx'\in\XX$
    $$\langle \xx'-\xx, \vv(\xx')-\vv(\xx)\rangle \leq -\lambda \|\vv(\xx')-\vv(\xx)\|^2.$$
\end{definition}
The sign convention in the definition above is natural for utility maximization. Equivalently, the operator $-\vv$ is $\lambda$-co-coercive. Before making further remarks on co-coercive games, we present the definition of a Nash equilibrium (NE).
\begin{definition}
    An action profile $\xxstar$ is a Nash equilibrium for game $\GG$ if for each player $n\in\NN$, $u_n(\xstar_n,\xxstar_{-n})\geq u_n(x_n,\xxstar_{-n})$ for all $x_n\in\RR^d$.
\end{definition}
Unlike strongly monotone games, where there is a unique NE, there can exist multiple or no NE in a co-coercive game. For co-coercive games with unconstrained action spaces, any Nash equilibrium $\xxstar$ satisfies the first-order condition $\vv(\xxstar)=0$ \cite{Zhou_2020}. We denote the set of Nash equilibria by
$$\Xstar=\left\{\xx\in\RR^{Nd}\mid \vv(\xx)=0\right\}.$$ As our objective in this paper is to study learning on NE in co-coercive games, we make the following assumption that the set of NE is non-empty.
\begin{assumption}\label{assu:non-empty}
    The set of NE, $\Xstar$, is non-empty.
\end{assumption}

\subsection{Examples of Co-coercive Games}
In this subsection, we present two examples of co-coercive games which are not strongly monotone. These help us better understand co-coercive games and build some intuition.

\subsubsection{Quadratic Games with NSD Interaction Matrix}

An interesting class of co-coercive games is given by quadratic games. For simplicity, we consider the scalar-action case, where each player chooses an action $x_n\in\RR$. The same construction extends naturally to the case of multi-dimensional actions by replacing scalar coefficients with block matrices.

Let $Q=(Q_{nm})\in\RR^{N\times N}$ be a symmetric negative semidefinite matrix, and define the utility of player $n$ as
$$
u_n(\xx)= \frac{1}{2}Q_{nn}x_n^2+\sum_{m\neq n}Q_{nm}x_nx_m.
$$
Then, the game gradient operator is linear and given by
$$
\vv(\xx)=Q\xx.
$$
Then the game is $\bigl(1/|\lambda_{\min}(Q)|\bigr)$-co-coercive, where $\lambda_{\min}(Q)$ denotes the smallest eigenvalue of $Q$ \cite[Pg.\ 79]{VI}. In contrast, the game would have been strongly monotone if $Q$ was a negative definite matrix.

An illustrative example is given by the following two-player non-separable game:
$$
u_1(x_1,x_2)= -\frac{1}{2}x_1^2-x_1x_2,
\quad
u_2(x_1,x_2)= -\frac{1}{2}x_2^2-x_1x_2.
$$
Here, the corresponding interaction matrix
$$
Q=
\begin{pmatrix}
-1 & -1\\
-1 & -1
\end{pmatrix}
$$
is negative semidefinite with eigenvalues $-2$ and $0$. Hence, the game is co-coercive but not strongly monotone. The set of NE for this game is given by
$$
\Xstar=
\{(x_1,x_2)\in\RR^2:x_1+x_2=0\},
$$
which illustrates that co-coercive games may admit infinitely many NE.

\subsubsection{Potential Games with Concave and Smooth Potential}
Another class of co-coercive games is given by potential games with a smooth concave potential function. A continuous game is called a \emph{potential game} if there exists a continuously differentiable function $\Phi:\XX\to\RR$ such that
$$
v_n(\xx)=\nabla_{x_n}\Phi(\xx),\qquad \forall n\in\NN.
$$
Equivalently, the game gradient operator satisfies
$$
\vv(\xx)=\nabla \Phi(\xx).
$$
Suppose that $\Phi$ is concave and has an $L$-Lipschitz gradient. Then, the game is $(1/L)$-co-coercive \cite[Theorem 2.1.5]{Nesterov}. In contrast, the potential game would have been strongly monotone if $\Phi$ was a strongly concave function.

An illustrative example is the following \emph{concave aggregate game}, where each player benefits from increasing their own action but incurs a cost that depends on the total aggregate action across all players. Let each player $n\in\NN$ choose an action $x_n\in\RR^d$, and consider utilities
\begin{equation}\label{eq:aggregate-utility}
    u_n(\xx)
    =
    \langle \phi,x_n\rangle
    -\frac{\gamma}{2}\left\|\sum_{m=1}^N x_m\right\|^2,
\end{equation}
where $\phi\in\RR^d$ is common across players and $\gamma>0$. The corresponding potential for the game is
\begin{equation}\label{eq:aggregate-potential}
    \Phi(\xx)
    =
    \sum_{n=1}^N \langle \phi,x_n\rangle
    -\frac{\gamma}{2}\left\|\sum_{n=1}^N x_n\right\|^2,
\end{equation}
which is concave and smooth with $L=\gamma N$. The set of NE for this game is given by
\[
\Xstar=
\left\{
\xx\in\RR^{Nd}
\;\middle|\;
\sum_{n=1}^N x_n=\frac{\phi}{\gamma}
\right\},
\]
which is an affine subspace of $\RR^{Nd}$.

\section{Stochastic Gradient Dynamics and Noise Model}\label{sec:SGD_noise}
We now introduce the stochastic gradient learning dynamics and the associated noise model considered throughout the paper. At time $t\in\{0,1,\ldots\}$, player $n$ plays action $x_{n,t}$ and obtains noisy feedback by receiving a noisy version $\hat{v}_{n,t}$ of the individual gradient evaluated at the current action profile. We denote by $\xx_t=(x_{1,t},\ldots, x_{N,t})$ and $\hat{\vv}_t=(\hat{v}_{1,t},\ldots, \hat{v}_{N,t})$, the joint action profile and the concatenated vector of gradient observations, respectively, at time $t$. 

Our players run stochastic gradient descent (SGD) using these noisy gradients to update their actions. Since players maximize utilities, this corresponds to a gradient ascent update. Formally, starting with some arbitrary profile $\xx_0$, each player $n$ performs the following update rule
\begin{equation}\label{SGD-iter}
    x_{n,t+1}=x_{n,t}+\beta_t\hat{v}_{n,t},
\end{equation}
where $\beta_t$ is the stepsize sequence which we discuss in detail in the next section. The noise in the gradient observation $\hat{v}_{n,t}$ can arise due to many reasons: for instance, estimates may be susceptible to random measurement errors, and the game's utility functions may be random themselves \cite{Zhou_2019}. We model this noise as a martingale difference sequence.
\begin{assumption}\label{assu:martingale}
    The noisy gradient $\hat{v}_{n,t}$ is of the form $\hat{v}_{n,t}=v_n(\xx_t)+M_{n,t+1}$. Let $\MM_t=(M_{1,t}, \ldots, M_{N,t})$. Then $\MM_{t+1}$ is a martingale difference sequence with respect to the filtration $\FF_t=\sigma(\xx_0, \MM_k, k\leq t)$, i.e., $\EE[\MM_{t+1}|\FF_t]=0$ for all $t\geq 0$. Moreover,
    $$\EE[\|\MM_{t+1}\|^2\mid\FF_t]\leq \sigma^2(1+\|\xx_t\|^2).$$
\end{assumption}
Our assumption on the scaling of the martingale noise is substantially more general than the assumptions used in prior work on co-coercive games. In particular, most of the results in \cite{Zhou_2020} are established under a \textit{relative noise model}, where $\EE[\|\MM_{t+1}\|^2\mid\FF_t]\leq \tau_t\|\vv(\xx_t)\|^2.$ This assumption significantly simplifies the analysis, since the noise magnitude vanishes as the iterates approach equilibrium, i.e., as $\|\vv(\xx_t)\|\to 0$. Additionally, \cite{Zhou_2020} assumes that $(\tau_t)$ is a decreasing sequence, which further facilitates convergence-rate analysis. They also study an \textit{absolute noise model} of the form $\EE[\|\MM_{t+1}\|^2\mid\FF_t]\leq \sigma_t^2.$ Under this model, only time-average guarantees are obtained in the constant-variance case $\sigma_t=\sigma$. In contrast, we consider the more general non-vanishing noise model, where the second moment of the noise scales affinely with the squared norm of the iterate. Allowing the second moment of the noise to grow with the iterates is important for modeling stochastic gradient estimators on unbounded action spaces, where uniformly bounded variance assumptions are often unrealistic.

\section{Main Results}\label{sec:results}

In this section, we study the convergence of the iterates $\xx_t$ generated by \eqref{SGD-iter} to the set of Nash equilibria $\Xstar$. Our focus is on last-iterate convergence rates. Since $\Xstar=\{\xx:\vv(\xx)\!=\!0\}$, we measure convergence through $\|\vv(\xx_t)\|$, which captures proximity to equilibrium. Our main objective is to obtain a last-iterate mean-square bound, i.e., a bound on $\EE[\|\vv(\xx_t)\|^2]$. As a consequence of our analysis, we also establish almost sure convergence of $\xx_t$ to $\Xstar$ and derive time-average guarantees. 

We first specify the class of admissible stepsize sequences. 
\begin{assumption}\label{assu:stepsize}
    The stepsize sequence is of the form $$\beta_t=\frac{1}{(t+T_0)^\bfrak},$$
    for some $T_0\geq \max\{1,\lambda^{-1/\bfrak},  (2\bfrak)^{1/(1-\bfrak)}\}$ and $0.5\!<\!\bfrak\!<\!1$.
\end{assumption}
This choice ensures that the stepsizes are non-summable ($\sum_t \beta_t\!=\!\infty)$, square-summable ($\sum_t \beta_t^2\!<\!\infty)$, and non-increasing. The assumption that $T_0$ is large enough is taken just to ensure that $\beta_t$ is smaller than $1$ and $\lambda$ and that $\beta_t\!-\!\beta_{t+1}\!\leq\! 0.5\beta_t^2$ for all $t\!\geq\!\! 0$. This assumption is not necessary; in its absence, our result holds for all $t$ greater than some $t_0$.

Our first result establishes almost sure convergence of the iterates generated by \eqref{SGD-iter} to the set of Nash equilibria.
\begin{theorem}\label{thm:as}
Suppose the players update their actions according to \eqref{SGD-iter} in a $\lambda$-co-coercive game, and that Assumptions \ref{assu:non-empty}-\ref{assu:stepsize} hold. Then the action profile $\xx_t$ converges to the set of Nash equilibria $\Xstar$ almost surely.
\end{theorem}
The proof is provided in Appendix \ref{app:proof:as}. Our next result provides a bound on the time-averaged residual.
\begin{theorem}\label{thm:time-avg}
    Suppose the players update their actions according to \eqref{SGD-iter} in a $\lambda$-co-coercive game, and that Assumptions \ref{assu:non-empty}-\ref{assu:stepsize} hold. Then there exists $C_1>0$ such that, for all $t\geq 0$,
    $$\frac{1}{t+1}\sum_{i=0}^t \EE\left[\|\vv(\xx_i)\|^2\right]\leq \frac{C_1}{(t+1)^{1-\bfrak}}.$$
\end{theorem}
Explicit value for $C_1$ has been provided along with the theorem's proof in Appendix \ref{app:proof:time-avg}. For the class of stepsize sequences in Assumption \ref{assu:stepsize}, the fastest time-averaged rate is achieved when $\bfrak=0.5+\delta$, where $\delta>0$ can be chosen arbitrarily small. This yields a rate of $\mathcal{O}(1/t^{(0.5-\delta)})$.

We now present our main result, which provides a last-iterate mean-square bound on the residual.
\begin{theorem}\label{thm:last-iterate}
    Suppose the players update their actions according to \eqref{SGD-iter} in a $\lambda$-co-coercive game, and that Assumptions \ref{assu:non-empty}-\ref{assu:stepsize} hold. Then there exist $C_2, C_3, C_4>0$ such that, for all $t \geq 0$, 
    \[
\EE\left[\|\vv(\xx_t)\|^2\right]
\leq
\begin{cases}
\frac{C_2}{(t+1)^{2\bfrak-1}},& \bfrak\in\bigl(\tfrac12,\tfrac23\bigr),\\[4pt]
\frac{C_3\log(t+1)}{(t+1)^{1/3}},& \bfrak=\tfrac23,\\[4pt]
\frac{C_4}{(t+1)^{1-\bfrak}},& \bfrak\in\bigl(\tfrac23,1\bigr).
\end{cases}
\]
\end{theorem}
The proof for this theorem has been provided in Appendix \ref{app:proof:last-iterate}. The best bound according to this result is obtained for $\bfrak=2/3$, yielding a rate of $\mathcal{O}(\log(t)/t^{1/3})$.

\section{Proof Outline}\label{sec:outline}
In this section, we present a proof sketch for the results in Section \ref{sec:results} through a series of lemmas. The proofs for these lemmas have been presented in Appendix \ref{app:proof_outline}. Our proof draws inspiration from the techniques used in the analysis of inexact Krasnosel'skii-Mann (KM) iterations \cite{Liang-inexact, Davis-splitting}, but requires additional arguments to handle the stochastic noise.

Recall that the update rule for player $n$ is given by $x_{n,t+1}=x_{n,t}+\beta_t\hat{v}_{n,t}$ where $\hat{v}_{n,t}=v_n(\xx_t)+M_{n,t+1}$. Combining these updates yields the following vector form
\begin{equation}\label{SGD-iter-vector}
    \xx_{t+1}=\xx_t+\beta_t(\vv(\xx_t)+\MM_{t+1}).
\end{equation}
Since the game $\GG$ is co-coercive, the operator $-\vv(\cdot)$ is co-coercive. Our first lemma obtains an intermediate bound on $\EE[\|\xx_t-\xxstar\|^2]$ where $\xxstar\in\Xstar$ is some NE for the game $\GG$.
\begin{lemma}\label{lemma:intermediate}
    Suppose Assumptions \ref{assu:non-empty}-\ref{assu:stepsize} are satisfied. Then for all $t\geq 0$, and for some $\xxstar\in\Xstar$,
    \begin{align*}
&\EE\left[\|\xx_{t+1}-\xxstar\|^2\mid \FF_t\right]\\
&\leq \left(1+2\beta_t^2\sigma^2\right)\|\xx_t-\xxstar\|^2-\lambda\beta_t\|\vv(\xx_t)\|^2\\
&\;\;+\beta_t^2\sigma^2(1+2\|\xxstar\|^2).
\end{align*}
\end{lemma}
The above lemma yields almost sure stability of the iterates $\{\xx_t\}$, i.e., almost sure boundedness of the iterates, by application of the Robbins-Siegmund Theorem \cite{Robbins-Siegmund} on the recursion. This is a key requirement for the ODE approach to convergence analysis of stochastic approximation (SA) algorithms \cite{Borkar-book}. The remainder of the proof for almost sure convergence follows by studying the ODE and using a Lyapunov function to show convergence to the set of NE.


Our next lemma shows that the iterates are bounded in expectation, which is essential to obtain a bound on the right hand side of the recursion.
\begin{lemma}\label{lemma:bounded_iter}
    Suppose Assumptions \ref{assu:non-empty}-\ref{assu:stepsize} hold. Then the iterates $\xx_t$ are bounded in expectation. Specifically, there exists $\Gamma_1>0$ such that for all $t\geq 0$, and for $\xxstar\in\Xstar$, 
    $$\EE\left[\|\xx_t-\xxstar\|^2\right]\leq \Gamma_1.$$
\end{lemma}

Our next lemma shows that the sum of weighted residual errors is bounded by a constant.
\begin{lemma}\label{lemma:sum_of_errors}
    Suppose Assumptions \ref{assu:non-empty}-\ref{assu:stepsize} hold. Then there exists $\Gamma_2$ such that for all $t\geq 0$,
    $$\sum_{i=0}^t \beta_i\EE\left[\|\vv(\xx_i)\|^2\right]\leq \Gamma_2.$$
\end{lemma}
The time-average guarantee (Theorem \ref{thm:time-avg}) follows directly from this lemma by simply noting that $\beta_i$ is non-increasing, which implies that
$$\beta_t\sum_{i=0}^t\EE\left[\|\vv(\xx_i)\|^2\right]\leq \sum_{i=0}^t\beta_i\EE\left[\|\vv(\xx_i)\|^2\right]\leq\Gamma_2.$$

But our primary goal is to obtain a last-iterate bound, i.e., obtain a bound on $\EE[\|\vv(\xx_t)\|^2]$. For this, we would like to get a bound on $\EE[\|\vv(\xx_i)\|^2]$ in terms of $\EE[\|\vv(\xx_t)\|^2]$. But the martingale difference noise $\MM_{t+1}$ does not allow us to obtain a useful bound directly. To solve this, we define an averaged error sequence $\UU_{t+1}=(1-\beta_t)\UU_t+\beta_t\MM_{t+1}$, with $\UU_0=0$, and a modified iterate $\zz_t=\xx_t-\UU_t$. We first show that $\EE[\|\xx_t-\zz_t\|^2]$ decays at a sufficiently fast rate, implying that we can analyze $\EE[\|\vv(\zz_t)\|^2]$ instead of $\EE[\|\vv(\xx_t)\|^2]$. The next step is studying the iterates $\zz_t$ and obtaining an iteration for these iterates. And finally, we fulfill the purpose of studying $\zz_t$ in terms of $\xx_t$, by obtaining a bound on $\EE[\|\vv(z_i)\|^2]$ in terms of $\EE[\|\vv(z_t)\|^2]$.
\begin{lemma}\label{lemma:noise-avg}
    Suppose Assumptions \ref{assu:non-empty}-\ref{assu:stepsize} hold.
    \begin{enumerate}[label=\alph*)]
        \item There exists $\Gamma_3>0$ such that for all $t\geq 0$,
        $$\EE\left[\|\xx_t-\zz_t\|^2\right]=\EE\left[\|\UU_t\|^2\right]\leq \Gamma_3\beta_t.$$
        \item There exists $\Gamma_4>0$ such that for all $t\geq 0$,
        $$\sum_{i=0}^t\beta_i\EE\left[\|\vv(\zz_i)\|^2\right]\leq \Gamma_4.$$
        \item For all $i\leq t-1$, there exists $\Gamma_5>0$ such that
        $$\EE\left[\|\vv(\zz_i)\|^2\right]\geq \EE\left[\|\vv(\zz_t)\|^2\right]-\Gamma_5\sum_{j=i}^{t-1}\beta_j^2.$$
    \end{enumerate}
\end{lemma}
The above lemma implies that 
$$\EE\left[\|\vv(\zz_t)\|^2\right]\sum_{i=0}^t\beta_i\leq \Gamma_4+\Gamma_5\sum_{i=0}^t\beta_i\sum_{j=i}^{t-1}\beta_j^2.$$
The summation on the right hand side is bounded for $\bfrak>2/3$, $\mathcal{O}(\log(t))$ for $\bfrak=2/3$, and $\mathcal{O}(t^{2-3b})$ for $1/2<b<2/3$. This gives us the following last-iterate bound for $\zz_t$.
\begin{lemma}\label{lemma:last-iter-bound-z}
    Suppose Assumptions \ref{assu:non-empty}-\ref{assu:stepsize} hold. Then there exists $\Gamma_6, \Gamma_7, \Gamma_8>0$ such that for all $t\geq 0$,
        \[
\EE\left[\|\vv(\zz_t)\|^2\right]
\leq
\begin{cases}
\frac{\Gamma_6}{(t+1)^{2\bfrak-1}},& \bfrak\in\bigl(\tfrac12,\tfrac23\bigr),\\[4pt]
\frac{\Gamma_7\log(t+1)}{(t+1)^{1/3}},& \bfrak=\tfrac23,\\[4pt]
\frac{\Gamma_8}{(t+1)^{1-\bfrak}},& \bfrak\in\bigl(\tfrac23,1\bigr).
\end{cases}
\]
\end{lemma}
The final step of the proof uses the Lipschitz nature of $\vv(\cdot)$ to bound $\EE[\|\vv(\xx_t)\|^2]$ in terms of $\EE[\|\vv(\zz_t)\|^2]$, completing the proof for the last-iterate bound.

\section{Conclusion}\label{sec:conc}
We studied last-iterate bounds for vanilla SGD in co-coercive games under noisy feedback. Specifically, we obtained the first last-iterate bound of $\mathcal{O}(\log(t)/t^{1/3})$ under the general noise model in which the variance of the noise can scale affinely with the squared norm of the iterates. This noise model subsumes the noise models in prior works which relied on vanishing noise for obtaining last-iterate bounds. As intermediate results, we also established almost sure convergence of the iterates to the Nash equilibrium set and derived time-average convergence guarantees.

Future directions include studying rates for co-coercive games when modified algorithms such as optimistic gradient descent or extragradient methods are used, and extending our results to the payoff-based feedback setting, where agents observe only their utility values.

\bibliographystyle{IEEEtran}
\bibliography{ref.bib}

\appendices
\section{Proofs for Theorems \ref{thm:as}, \ref{thm:time-avg}, \ref{thm:last-iterate}}

\subsection{\textbf{Proof for Theorem \ref{thm:as}}}\label{app:proof:as}
\begin{proof}
Let us first analyze solutions of the following ODE.
\begin{equation}\label{ODE}
    \dot{\xx}(t)=\vv(\xx(t)).
\end{equation}
Define $V(\xx)=\|\xx-\xxstar\|^2$ for some $\xxstar\in\Xstar$. Then we have that $V(\xx)\geq 0$ for all $\xx\in\RR^{Nd}$ and $V(\xx)\rightarrow\infty$ as $\xx\rightarrow\infty$. Also, note that 
\begin{align*}
    \langle \nabla V(\xx), \vv(\xx)\rangle &= 2\langle \xx-\xxstar, \vv(\xx)\rangle\\
    &=2\langle \xx-\xxstar, \vv(\xx)-\vv(\xxstar)\rangle\\
    &\leq -2\lambda\|\vv(\xx)-\vv(\xxstar)\|^2\\
    &=-2\lambda\|\vv(\xx)\|^2\leq 0,
\end{align*}
with equality iff $\vv(\xx)=0$. Here, the second and last equality follow from the fact that $\vv(\xxstar)=0$, and the inequality follows from the co-coercivity of the operator $-\vv(\cdot)$. Then by LaSalle's theorem or invariance principle \cite[Theorem 3.5]{slotine1991applied}, all solutions of the ODE globally asymptotically converge to the set $\{\xx|\langle \nabla V(\xx), \vv(\xx)\rangle=0\}=\{\xx|\vv(\xx)=0\}$ which is precisely $\Xstar$, the set of NE. 

Now, we have to relate the iteration \eqref{SGD-iter-vector} with the ODE and show that the iterates also converge to the set. For this, we use the ODE approach for the analysis of stochastic approximation algorithms. We have to first show that the iterates are almost surely bounded. For this, we first define $V_t=\|\xx_t-\xxstar\|^2$, $A_t=2\beta_t^2\sigma^2$, $B_t=\lambda \beta_t\|\vv(\xx_t)\|^2$, and $C_t=\beta_t^2\sigma^2(1+2\|\xxstar\|^2)$. Using Lemma \ref{lemma:intermediate}, we get 
$$\EE[V_{t+1}|\FF_t]\leq (1+A_t)V_t-B_t+C_t.$$
Since $\beta_t$ is square-summable, we have $\sum_{t}A_t<\infty$ and $\sum_t C_t<\infty$. Then, we can apply Robbins-Siegmund Theorem \cite{Robbins-Siegmund} which says that $V_t$ converges to a finite random variable. In particular, $\sup_t V_t<\infty$ almost surely, and hence
\[
\sup_t \|\xx_t\|<\infty
\qquad\text{a.s.}
\]
Thus, the iterates are almost surely bounded. Now, all assumptions required for \cite[Theorem 2.1]{Borkar-book} are satisfied (the map $\vv(\cdot)$ is Lipschitz, the stepsize sequence is non-summable but square-summable, the noise is martingale difference, and the iterates are almost surely bounded). Note that \eqref{ODE} is the limiting ODE for \eqref{SGD-iter-vector}. Since all solutions of the ODE \eqref{ODE} converge to the set $\Xstar$, the iterates $\xx_t$ converge to the set $\Xstar$ almost surely \cite[Theorem 2.1]{Borkar-book}.
\end{proof}

\subsection{\textbf{Proof for Theorem \ref{thm:time-avg}}}\label{app:proof:time-avg}
\begin{proof}
By Lemma~\ref{lemma:sum_of_errors},
\[
\sum_{i=0}^t \beta_i\EE\left[\|\vv(\xx_i)\|^2\right]\le \Gamma_2,
\qquad \forall t\ge 0.
\]
Since $\{\beta_t\}$ is non-increasing,
\[
\beta_t\sum_{i=0}^t \EE\left[\|\vv(\xx_i)\|^2\right]
\le
\sum_{i=0}^t \beta_i\EE\left[\|\vv(\xx_i)\|^2\right]
\le
\Gamma_2.
\]
Therefore,
\begin{align*}
    \frac{1}{t+1}\sum_{i=0}^t \EE\left[\|\vv(\xx_i)\|^2\right]&\leq \frac{\Gamma_2}{(t+1)\beta_t}\\
    &=\Gamma_2\frac{(t+T_0)^\bfrak}{t+1}\\
    &\leq \Gamma_2T_0^{\bfrak}\frac{1}{(t+1)^{1-\bfrak}}.
\end{align*}
Here the last step follows from the fact that $t+T_0\leq T_0(t+1)$. This completes the proof with $C_1=\Gamma_2T_0^{\bfrak}$.
\end{proof}

\subsection{\textbf{Proof for Theorem \ref{thm:last-iterate}}}\label{app:proof:last-iterate}
\begin{proof}
Note that $\vv(\cdot)$ is $(1/\lambda)$-Lipschitz (follows from the co-coercive nature of $-\vv(\cdot)$). Then,
\[
\|\vv(\xx_t)\|
\le
\|\vv(\zz_t)\|+\|\vv(\xx_t)-\vv(\zz_t)\|
\le
\|\vv(\zz_t)\|+\frac{1}{\lambda}\|\UU_t\|.
\]
Hence
\[
\|\vv(\xx_t)\|^2
\le
2\|\vv(\zz_t)\|^2+\frac{2}{\lambda^2}\|\UU_t\|^2.
\]
Taking expectation and using Lemma~\ref{lemma:noise-avg}(a), we get
\[
\EE\left[\|\vv(\xx_t)\|^2\right]
\le
2\EE\left[\|\vv(\zz_t)\|^2\right]
+\frac{2\Gamma_3}{\lambda^2}\beta_t.
\]
Recall that
        \[
\EE\left[\|\vv(\zz_t)\|^2\right]
\leq
\begin{cases}
\frac{\Gamma_6}{(t+1)^{2\bfrak-1}},& \bfrak\in\bigl(\tfrac12,\tfrac23\bigr),\\[4pt]
\frac{\Gamma_7\log(t+1)}{(t+1)^{1/3}},& \bfrak=\tfrac23,\\[4pt]
\frac{\Gamma_8}{(t+1)^{1-\bfrak}},& \bfrak\in\bigl(\tfrac23,1\bigr).
\end{cases}
\]
For $\bfrak\in(1/2,2/3)$, 
$$\beta_t=(t+T_0)^{-\bfrak}\leq (t+1)^{-\bfrak}\leq (t+1)^{-(2\bfrak-1)}.$$
Hence
$$\EE\left[\|\vv(\xx_t)\|^2\right]\leq \left(2\Gamma_6+\frac{2\Gamma_3}{\lambda^2}\right)\frac{1}{(t+1)^{2\bfrak-1}}.$$
This completes the first case with $C_2=2\Gamma_6+2\Gamma_3/\lambda^2$. Similarly, for the other two cases, it can be verified that $\EE[\|\vv(\zz_t)\|^2]$ orderwise dominates the term $\beta_t$. Hence, for $\bfrak=2/3$,
$$\EE\left[\|\vv(\xx_t)\|^2\right]\leq C_3\frac{\log(t+1)}{(t+1)^{1/3}},$$
where $C_3=2\Gamma_7+2\Gamma_3/\lambda^2$, and for  $\bfrak\in(2/3,1)$,
$$\EE\left[\|\vv(\xx_t)\|^2\right]\leq \frac{C_4}{(t+1)^{1-\bfrak}},$$
where $C_4=2\Gamma_8+2\Gamma_3/\lambda^2$.
\end{proof}

\newpage
\onecolumn
\section{Proofs from Section \ref{sec:outline}}\label{app:proof_outline}
\subsection{\textbf{Proof for Lemma \ref{lemma:intermediate}}}
\begin{proof}
Let $\xxstar\in\Xstar$. Then, for all \(t\geq 0\),
\begin{align*}
\xx_{t+1}-\xxstar
&=\xx_t-\xxstar+\beta_t(\vv(\xx_t)-\vv(\xxstar)+\MM_{t+1}).
\end{align*}
Now, using the standard expansion, we get
\begin{subequations}\label{split1}
\begin{align}
\|\xx_{t+1}-\xxstar\|^2
&=\|\xx_t-\xxstar+\beta_t(\vv(\xx_t)-\vv(\xxstar)+\MM_{t+1})\|^2\nonumber\\
&=\|\xx_t-\xxstar\|^2
+\beta_t^2\|\vv(\xx_t)-\vv(\xxstar)+\MM_{t+1}\|^2\label{split11}\\
&\;\;+2\beta_t\langle \xx_t-\xxstar,\vv(\xx_t)-\vv(\xxstar)+\MM_{t+1}\rangle.\label{split12}
\end{align}
\end{subequations}
For the term \eqref{split11}, we note that
\begin{align*}
\|\vv(\xx_t)-\vv(\xxstar)+\MM_{t+1}\|^2
&=\|\vv(\xx_t)-\vv(\xxstar)\|^2+\|\MM_{t+1}\|^2
+2\langle \vv(\xx_t)-\vv(\xxstar),\MM_{t+1}\rangle.
\end{align*}
For the term \eqref{split12}, since \(-\vv(\cdot)\) is \(\lambda\)-co-coercive, we have
\begin{align*}
\langle \xx_t-\xxstar,\vv(\xx_t)-\vv(\xxstar)\rangle
\leq -\lambda\|\vv(\xx_t)-\vv(\xxstar)\|^2.
\end{align*}
Hence,
\begin{align*}
2\beta_t\langle \xx_t-\xxstar,\vv(\xx_t)-\vv(\xxstar)+\MM_{t+1}\rangle
&\leq -2\lambda\beta_t\|\vv(\xx_t)-\vv(\xxstar)\|^2
+2\beta_t\langle \xx_t-\xxstar,\MM_{t+1}\rangle.
\end{align*}
Returning to \eqref{split1}, we get
\begin{align*}
\|\xx_{t+1}-\xxstar\|^2
&\leq \|\xx_t-\xxstar\|^2
+\beta_t^2\|\vv(\xx_t)-\vv(\xxstar)\|^2
+\beta_t^2\|\MM_{t+1}\|^2\\
&\;\;+2\beta_t^2\langle \vv(\xx_t)-\vv(\xxstar),\MM_{t+1}\rangle
-2\lambda\beta_t\|\vv(\xx_t)-\vv(\xxstar)\|^2+2\beta_t\langle \xx_t-\xxstar,\MM_{t+1}\rangle\\
&= \|\xx_t-\xxstar\|^2
-\beta_t(2\lambda-\beta_t)\|\vv(\xx_t)-\vv(\xxstar)\|^2
+\beta_t^2\|\MM_{t+1}\|^2\\
&\;\;+2\beta_t\langle \xx_t-\xxstar,\MM_{t+1}\rangle
+2\beta_t^2\langle \vv(\xx_t)-\vv(\xxstar),\MM_{t+1}\rangle\\
&\leq \|\xx_t-\xxstar\|^2
-\lambda\beta_t\|\vv(\xx_t)-\vv(\xxstar)\|^2
+\beta_t^2\|\MM_{t+1}\|^2\\
&\;\;+2\beta_t\langle \xx_t-\xxstar,\MM_{t+1}\rangle
+2\beta_t^2\langle \vv(\xx_t)-\vv(\xxstar),\MM_{t+1}\rangle.
\end{align*}
Here the final inequality follows from the fact that \(\beta_t\leq \lambda\) which implies that \(2\lambda-\beta_t\geq \lambda\) for all $t$. On taking expectation conditioned on \(\FF_t\), we have
\[
\EE[\langle \xx_t-\xxstar,\MM_{t+1}\rangle\mid \FF_t]=0,
\qquad
\EE[\langle \vv(\xx_t)-\vv(\xxstar),\MM_{t+1}\rangle\mid \FF_t]=0,
\]
since \(\{\MM_{t+1}\}\) is a martingale difference sequence with respect to \(\{\FF_t\}\), and both \(\xx_t\) and \(\vv(\xx_t)\) are \(\FF_t\)-measurable.
Therefore,
\begin{align*}
\EE\left[\|\xx_{t+1}-\xxstar\|^2\mid\FF_t\right]
&\leq \|\xx_t-\xxstar\|^2
+\beta_t^2\EE\left[\|\MM_{t+1}\|^2\mid\FF_t\right]
-\lambda\beta_t\|\vv(\xx_t)-\vv(\xxstar)\|^2\\
&\leq \|\xx_t-\xxstar\|^2
+\beta_t^2\sigma^2(1+\|\xx_t\|^2)
-\lambda\beta_t\|\vv(\xx_t)\|^2\\
&\leq \|\xx_t-\xxstar\|^2
+2\beta_t^2\sigma^2\|\xx_t-\xxstar\|^2+\beta_t^2\sigma^2(1+2\|\xxstar\|^2)
-\lambda\beta_t\|\vv(\xx_t)\|^2,
\end{align*}
where we used \(\vv(\xxstar)=0\) for the second inequality. This completes the proof for Lemma \ref{lemma:intermediate}.

\end{proof}

\subsection{\textbf{Proof for Lemma \ref{lemma:bounded_iter}}}
\begin{proof}
Taking expectation and dropping the negative residual term
in Lemma~\ref{lemma:intermediate} gives, for all $t\ge 0$,
$$\EE\left[\|\xx_{t+1}-\xxstar\|^2\right]\leq \left(1+2\beta_t^2\sigma^2\right)\EE\left[\|\xx_t-\xxstar\|^2\right]+\beta_t^2\sigma^2(1+2\|\xxstar\|^2).$$
Expanding the recursion, we get the following.
\begin{align*}
    \EE\left[\|\xx_{t}-\xxstar\|^2\right]&\leq \|\xx_0-\xxstar\|^2\prod_{i=0}^{t-1}\left(1+2\beta_i^2\sigma^2\right)+\sum_{i=0}^{t-1}\sigma^2\beta_i^2(1+2\|\xxstar\|^2)\prod_{j=i+1}^{t-1}\left(1+2\sigma^2\beta_j^2\right)\\
    &\leq e^{2\sigma^2\sum_{i=0}^{\infty}\beta_i^2}\left(\|\xx_0-\xxstar\|^2+\sigma^2\left(1+2\|\xxstar\|^2\right)\sum_{i=0}^{t-1}\beta_i^2\right)
\end{align*}
Let $D_1\coloneqq\sum_{i=0}^\infty \beta_i^2<\infty$. Then, 
\begin{align*}
    \EE\left[\|\xx_{t}-\xxstar\|^2\right]\leq e^{2\sigma^2D_1}\left(\|\xx_0-\xxstar\|^2+\sigma^2\left(1+2\|\xxstar\|^2\right)D_1\right)\eqqcolon \Gamma_1.
\end{align*}
Hence $\EE\left[\|\xx_{t}-\xxstar\|^2\right]\leq \Gamma_1$ for all $t\geq 0$.
\end{proof} 

\subsection{\textbf{Proof for Lemma \ref{lemma:sum_of_errors}}}
\begin{proof}
Taking expectation in Lemma~\ref{lemma:intermediate}, and rearranging terms, we get
\[
\lambda\beta_i\EE\left[\|\vv(\xx_i)\|^2\right]
\le
\EE\left[\|\xx_i-\xxstar\|^2\right]-\EE\left[\|\xx_{i+1}-\xxstar\|^2\right]
+
\beta_i^2\sigma^2\bigl(1+2\|\xxstar\|^2\bigr)
+
2\beta_i^2\sigma^2\EE\left[\|\xx_i-\xxstar\|^2\right].
\]
Summing from \(i=0\) to \(t\) gives
\[
\lambda\sum_{i=0}^{t}\beta_i\EE\left[\|\vv(\xx_i)\|^2\right]
\le
\EE\left[\|\xx_0-\xxstar\|^2\right]
-
\EE\left[\|\xx_{t+1}-\xxstar\|^2\right]
+
\sigma^2\bigl(1+2\|\xxstar\|^2\bigr)\sum_{i=0}^{t}\beta_i^2
+
2\sigma^2\sum_{i=0}^{t}\beta_i^2\EE\left[\|\xx_i-\xxstar\|^2\right].
\]
Dropping the nonnegative term \(\EE[\|\xx_{t+1}-\xxstar\|^2]\), and using the fact that $\sum_t\beta_t^2=D_1$ along with Lemma \ref{lemma:bounded_iter}, we obtain
\[
\lambda\sum_{i=0}^{t}\beta_i\EE\left[\|\vv(\xx_i)\|^2\right]
\le
\EE\left[\|\xx_0-\xxstar\|^2\right]
+
\sigma^2\bigl(1+2\|\xxstar\|^2\bigr)D_1
+
2\sigma^2D_1\Gamma_1.
\]
This completes the proof for Lemma \ref{lemma:sum_of_errors} with $\Gamma_2\coloneqq \frac{1}{\lambda}\left(\|\xx_0-\xxstar\|^2+\sigma^2D_1\left(1+2\|\xxstar\|^2+2\Gamma_1\right)\right)$
\end{proof}

\subsection{\textbf{Proof for Lemma \ref{lemma:noise-avg}}}
\begin{proof}
We first show that $\vv$ is Lipschitz. For all $\xx,\xx'\in\RR^{Nd}$,
\begin{align}
\lambda\|\vv(\xx)-\vv(\xx')\|^2 &\leq-\langle\xx-\xx',\vv(\xx)-\vv(\xx')\rangle \nonumber\\
&\leq\|\xx-\xx'\|\|\vv(\xx)-\vv(\xx')\|\nonumber\\
\implies\|\vv(\xx)-\vv(\xx')\| &\leq\frac{1}{\lambda}\|\xx-\xx'\|.\label{eq:v_lip}
\end{align}

\noindent\textbf{(a) Bound on $\EE\left[\|\UU_t\|^2\right]$.}
We first note that $\EE[\|\MM_{t+1}\|^2]$ is bounded by a constant for all $t\geq 0$.
\begin{align*}
\EE\left[\|\MM_{t+1}\|^2\right]&\leq\sigma^2(1+\EE\left[\|\xx_t\|^2\right])\nonumber\\
&\leq\sigma^2(1+2\EE\left[\|\xx_t-\xxstar\|^2\right]+2\|\xxstar\|^2)\nonumber\\
&\leq\sigma^2(1+2\Gamma_1+2\|\xxstar\|^2)\eqqcolon D_2.
\end{align*}
Now, we note that 
$$\UU_{t}=\sum_{i=0}^{t-1}\beta_i\prod_{j=i+1}^{t-1}(1-\beta_j)\MM_{i+1}.$$
Using the property that the terms of a martingale difference sequence are orthogonal,
\begin{align*}
    \EE\left[\|\UU_t\|^2\right]&\leq \sum_{i=0}^{t-1} \left(\beta_i\prod_{j=i+1}^{t-1}(1-\beta_j)\right)^2\EE\left[\|\MM_{i+1}\|^2\right]\nonumber\leq D_2\sum_{i=0}^{t-1} \beta_i^2\prod_{j=i+1}^{t-1}(1-\beta_j)
\end{align*}

Define $S_t=\sum_{i=0}^{t-1} \beta_i^2\prod_{j=i+1}^{t-1}(1-\beta_j)$. We next show that $S_t\leq 2\beta_t$ for all $t\geq 0$. For this, we first note that,
\begin{align*}
\beta_t-\beta_{t+1}&=\frac{1}{(t+T_0)^\bfrak}-\frac{1}{(t+1+T_0)^\bfrak}\\
&=\frac{1}{(t+T_0)^\bfrak}\left(1-\left(1+\frac{1}{t+T_0}\right)^{-\bfrak}\right)\\
&\leq\frac{1}{(t+T_0)^\bfrak}\frac{\bfrak}{t+T_0}=\frac{\bfrak}{(t+T_0)^{\bfrak+1}}.
\end{align*}
Here the inequality follows from the fact that $(1+x)^{-b}\geq 1-bx$, where $b\in(0,1)$ and $x>0$. Then, 
\begin{align*}
    \frac{\beta_t-\beta_{t+1}}{\beta_t^2}\leq \frac{\bfrak}{(t+T_0)^{1-\bfrak}}\leq \frac{1}{2},
\end{align*}
where the final inequality follows from our assumption that $T_0\geq (2\bfrak)^{1/(1-\bfrak)}$. This implies that $\beta_t-\beta_{t+1}\leq \frac{1}{2}\beta_t^2$ for all $t\geq 0$. 

Now, suppose that $S_t\leq 2\beta_t$ for some $t\geq0$. Then,
\begin{align*}
    S_{t+1}&=(1-\beta_t)S_t+\beta_t^2\leq 2\beta_t(1-\beta_t)+\beta_t^2= 2\beta_t-\beta_t^2\leq 2\beta_{t+1}.
\end{align*}
Note that $S_0\leq 2\beta_0$ and hence, by induction, $S_t\leq 2\beta_t$ for all $t\geq 0$. This completes the proof for Lemma \ref{lemma:noise-avg}(a) with $\Gamma_3=2D_2$.

\noindent\textbf{(b) Bound on $\sum_{i=0}^{t}\beta_i\EE\left[\|\vv(\zz_i)\|^2\right]$.}
Using $(1/\lambda)$-Lipschitz property of the map $\vv(\cdot)$ \eqref{eq:v_lip}, we get 
\begin{align*}
\|\vv(\zz_t)\|&\leq\|\vv(\xx_t)\|+\|\vv(\zz_t)-\vv(\xx_t)\|\\
&\leq\|\vv(\xx_t)\|+\frac{1}{\lambda}\|\zz_t-\xx_t\|\\
&=\|\vv(\xx_t)\|+\frac{1}{\lambda}\|\UU_t\|\\
\implies\|\vv(\zz_t)\|^2 &\leq2\|\vv(\xx_t)\|^2+\frac{2}{\lambda^2}\|\UU_t\|^2.
\end{align*}
Then using Lemma \ref{lemma:noise-avg}(a), 
\begin{align*}
\sum_{i=0}^{t}\beta_i\EE\left[\|\vv(\zz_i)\|^2\right]&\leq2\sum_{i=0}^{t}\beta_i\EE\left[\|\vv(\xx_i)\|^2\right]+\frac{2}{\lambda^2}\sum_{i=0}^{t}\beta_i\EE\left[\|\UU_i\|^2\right]\\
&\leq2\Gamma_2+\frac{2\Gamma_3}{\lambda^2}\sum_{i=0}^{t}\beta_i^2\\
&\leq2\Gamma_2+\frac{2\Gamma_3}{\lambda^2}D_1.
\end{align*}
This completes the proof for part (b) with $\Gamma_4=2\Gamma_2+2\Gamma_3D_1/\lambda^2$.

\noindent\textbf{(c) Bound on $\EE[\|\vv(\zz_i)\|^2]$ in terms of $\EE[\|\vv(\zz_t)\|^2]$.}

From the recursions of $\xx_t$ and $\UU_t$,
\begin{align*}
\zz_{t+1}&=\xx_{t+1}-\UU_{t+1}\\
&=\xx_t+\beta_t(\vv(\xx_t)+\MM_{t+1})-\left((1-\beta_t)\UU_t+\beta_t\MM_{t+1}\right)\\
&=\zz_t+\beta_t(\vv(\xx_t)+\UU_t)\\
&=\zz_t+\beta_t(\vv(\zz_t)+\ee_t),
\end{align*}
where $\ee_t \coloneqq\vv(\xx_t)-\vv(\zz_t)+\UU_t$. Using \eqref{eq:v_lip}, we have $\|\ee_t\|\leq \|\vv(\xx_t)-\vv(\zz_t)\|+\|\UU_t\|\leq (1+1/\lambda)\|\UU_t\|$. This gives us 
\begin{align}
\implies\EE\left[\|\ee_t\|^2\right] &\leq\left(1+\frac{1}{\lambda}\right)^2\EE\left[\|\UU_t\|^2\right]\leq\left(1+\frac{1}{\lambda}\right)^2\Gamma_3\beta_t.\label{eq:e_bound}
\end{align}
Next, note that
\begin{align*}
\zz_{t+1}-\zz_t&=\beta_t(\vv(\zz_t)+\ee_t)\\
\implies\vv(\zz_t)&=\frac{\zz_{t+1}-\zz_t-\beta_t\ee_t}{\beta_t}.
\end{align*}
Therefore,
\begin{align}\label{eqn-v-split1}
\|\vv(\zz_{t+1})\|^2&=\|\vv(\zz_t)\|^2+\|\vv(\zz_{t+1})-\vv(\zz_t)\|^2+2\langle\vv(\zz_{t+1})-\vv(\zz_t),\vv(\zz_t)\rangle\nonumber\\
&=\|\vv(\zz_t)\|^2+\|\vv(\zz_{t+1})-\vv(\zz_t)\|^2+\frac{2}{\beta_t}\langle\vv(\zz_{t+1})-\vv(\zz_t),\zz_{t+1}-\zz_t-\beta_t\ee_t\rangle\nonumber\\
&\leq \|\vv(\zz_t)\|^2+\|\vv(\zz_{t+1})-\vv(\zz_t)\|^2-\frac{2\lambda}{\beta_t}\|\vv(\zz_{t+1})-\vv(\zz_t)\|^2-2\langle\vv(\zz_{t+1})-\vv(\zz_t),\ee_t\rangle\nonumber\\
&=\|\vv(\zz_t)\|^2-\left(\frac{2\lambda}{\beta_t}-1\right)\|\vv(\zz_{t+1})-\vv(\zz_t)\|^2-2\langle\vv(\zz_{t+1})-\vv(\zz_t),\ee_t\rangle.
\end{align}
Here the inequality follows from the fact that map $-\vv(\cdot)$ is $\lambda$-co-coercive.

For any $c>0$, we have,
\begin{align*}
0\leq\left\|\sqrt{c}u+\frac{1}{\sqrt{c}}w\right\|^2=c\|u\|^2+2\langle u,w\rangle+\frac{1}{c}\|w\|^2.
\end{align*}
Hence, 
 \begin{align*}
     -c\|u\|^2-2\langle u,w\rangle &\leq\frac{1}{c}\|w\|^2
 \end{align*}
 Taking $c=\frac{2\lambda}{\beta_t}-1$, $u=\vv(\zz_{t+1})-\vv(\zz_t)$, and $w=\ee_t$,
\begin{align*}
    -\left(\frac{2\lambda}{\beta_t}-1\right)\|\vv(\zz_{t+1})-\vv(\zz_t)\|^2-2\langle\vv(\zz_{t+1})-\vv(\zz_t),\ee_t\rangle\leq \frac{1}{\left(\frac{2\lambda}{\beta_t}-1\right)}\|\ee_t\|^2\leq \frac{\beta_t}{\lambda}\|\ee_t\|^2.
\end{align*}
Here the final inequality follows from the assumption that $\beta_t\leq \lambda$ for all $t\geq 0$. Combining this with \eqref{eqn-v-split1}, we get
\begin{align*}
\|\vv(\zz_{t+1})\|^2\leq\|\vv(\zz_t)\|^2+\frac{\beta_t}{\lambda}\|\ee_t\|^2.
\end{align*}
This implies that 
\begin{align*}
\EE\left[\|\vv(\zz_{t+1})\|^2\right]\leq\EE\left[\|\vv(\zz_t)\|^2\right]+\Gamma_5\beta_t^2,
\end{align*}
where $\Gamma_5=(1+1/\lambda)^2\Gamma_3/\lambda$. This follows from the bound on $\EE[\|\ee_t\|^2]$ in \eqref{eq:e_bound}.
Here, the second inequality uses $\beta_t\leq\lambda$, and the last inequality uses \eqref{eq:e_bound}.

Iterating from $j=i$ to $t-1$ gives
\begin{align*}
\EE\left[\|\vv(\zz_i)\|^2\right]\geq\EE\left[\|\vv(\zz_t)\|^2\right]
-\Gamma_5\sum_{j=i}^{t-1}\beta_j^2.
\end{align*}
This completes the proof for Lemma \ref{lemma:noise-avg}.
\end{proof}

\subsection{\textbf{Proof for Lemma \ref{lemma:last-iter-bound-z}}}
\begin{proof}
For each $i\leq t-1$, Lemma~\ref{lemma:noise-avg}(c) gives
\begin{align*}
\EE\left[\|\vv(\zz_t)\|^2\right]&\leq\EE\left[\|\vv(\zz_i)\|^2\right]+\Gamma_5\sum_{j=i}^{t-1}\beta_j^2.
\end{align*}
Multiplying by $\beta_i$ and summing from $i=0$ to $t-1$,
\begin{align*}
\EE\left[\|\vv(\zz_t)\|^2\right]\sum_{i=0}^{t-1}\beta_i&\leq\sum_{i=0}^{t-1}\beta_i\EE\left[\|\vv(\zz_i)\|^2\right]+\Gamma_5\sum_{i=0}^{t-1}\beta_i\sum_{j=i}^{t-1}\beta_j^2.
\end{align*}
Adding $\beta_t\EE\left[\|\vv(\zz_t)\|^2\right]$ to both sides and using Lemma~\ref{lemma:noise-avg}(b),
\begin{align*}
\EE\left[\|\vv(\zz_t)\|^2\right]\sum_{i=0}^{t}\beta_i&\leq\Gamma_4+\Gamma_5\sum_{i=0}^{t-1}\beta_i\sum_{j=i}^{t-1}\beta_j^2\\
\implies\EE\left[\|\vv(\zz_t)\|^2\right]&\leq\frac{\Gamma_4}{\sum_{i=0}^{t}\beta_i}+\Gamma_5\frac{\sum_{i=0}^{t-1}\beta_i\sum_{j=i}^{t-1}\beta_j^2}{\sum_{i=0}^{t}\beta_i}.
\end{align*}
To bound the first term, we use the same simplification as the proof for Theorem \ref{thm:time-avg} to get
$$\frac{\Gamma_4}{\sum_{i=0}^{t}\beta_i}\leq \frac{\Gamma_4T_0^\bfrak}{(t+1)^{1-\bfrak}}.$$
Now, to bound the second term, note that
\begin{align*}
\sum_{i=0}^{j}\beta_i\leq\sum_{i=0}^{j}(i+1)^{-\bfrak}\leq1+\int_1^{j+1}x^{-\bfrak}dx\leq\frac{(j+1)^{1-\bfrak}}{1-\bfrak}.
\end{align*}
Now, by swapping the order of summation and using the above bound, we get
\begin{align*}
\sum_{i=0}^{t-1}\beta_i\sum_{j=i}^{t-1}\beta_j^2 =\sum_{j=0}^{t-1}\left(\sum_{i=0}^{j}\beta_i\right)\beta_j^2\leq\frac{1}{1-\bfrak}\sum_{j=0}^{t-1}(j+1)^{1-3\bfrak}.
\end{align*}
Using standard integral comparison, we bound the above sum as
\begin{align*}
\sum_{j=0}^{t-1}(j+1)^{1-3\bfrak}
\leq
\begin{cases}
\dfrac{(t+1)^{2-3\bfrak}}{2-3\bfrak},&\bfrak\in(\tfrac12,\tfrac23),\\
2\log(t+1),&\bfrak=\tfrac23,\\
\displaystyle\sum_{j=0}^{\infty}(j+1)^{1-3\bfrak},&\bfrak\in(\tfrac23,1).
\end{cases}
\end{align*}
Substituting this and using the lower bound on $\sum_{i=0}^{t}\beta_i$, we get
\begin{align*}
\Gamma_5\frac{\sum_{i=0}^{t-1}\beta_i\sum_{j=i}^{t-1}\beta_j^2}{\sum_{i=0}^{t}\beta_i}
\leq
\begin{cases}
\dfrac{\Gamma_5T_0^\bfrak}{(1-\bfrak)(2-3\bfrak)}(t+1)^{-(2\bfrak-1)},&\bfrak\in(\tfrac12,\tfrac23),\\
6\Gamma_5T_0^{2/3}\dfrac{\log(t+1)}{(t+1)^{1/3}},&\bfrak=\tfrac23,\\
\dfrac{\Gamma_5T_0^\bfrak}{1-\bfrak}\left(\sum_{j=0}^{\infty}(j+1)^{1-3\bfrak}\right)(t+1)^{-(1-\bfrak)},&\bfrak\in(\tfrac23,1).
\end{cases}
\end{align*}
Using the upper bound on the first term and the above upper bound on the second term, we obtain
\begin{align*}
\EE\left[\|\vv(\zz_t)\|^2\right]
\leq
\begin{cases}
\Gamma_6(t+1)^{-(2\bfrak-1)},&\bfrak\in(\tfrac12,\tfrac23),\\
\Gamma_7\dfrac{\log(t+1)}{(t+1)^{1/3}},&\bfrak=\tfrac23,\\
\Gamma_8(t+1)^{-(1-\bfrak)},&\bfrak\in(\tfrac23,1).
\end{cases}
\end{align*}
Here
\begin{align*}
\Gamma_6&\coloneqq\Gamma_4T_0^\bfrak+\frac{\Gamma_5T_0^\bfrak}{(1-\bfrak)(2-3\bfrak)},\\
\Gamma_7&\coloneqq\frac{\Gamma_4T_0^{2/3}}{\log2}+6\Gamma_5T_0^{2/3},\\
\Gamma_8&\coloneqq\Gamma_4T_0^\bfrak+\frac{\Gamma_5T_0^\bfrak}{1-\bfrak}\sum_{j=0}^{\infty}(j+1)^{1-3\bfrak}.
\end{align*}
This completes the proof for Lemma 5.
\end{proof}
\end{document}